\documentclass[10pt]{article}
\usepackage{amsmath,amssymb,amsfonts}
\usepackage{algorithm}
\usepackage{algorithmic}
\usepackage{theorem}
\usepackage{latexsym}
\usepackage{graphicx}
\usepackage{textcomp}
\usepackage{defs}
\def\qed{\hfill$\Box$\par\vskip1em}
\newtheorem{theorem}{Theorem}
\newtheorem{definition}{Definition}
\newtheorem{lemma}[theorem]{Lemma}

\newenvironment{proof}{\textbf{Proof.} } 

\title{An asynchronous message-passing distributed algorithm for the global critical section problem\thanks{This is a modified version of the conference paper in PDAA2017.}
}

\author{Sayaka Kamei\thanks{Dept. of Information Engineering, Graduate School of Engineering, Hiroshima University, s-kamei@se.hiroshima-u.ac.jp}\and
Hirotsugu Kakugawa\thanks{Dept. of Computer Science, Graduate School of Information Science and Technology, Osaka University, kakugawa@ist.osaka-u.ac.jp}}

\date{}
\begin{document}
\maketitle

\begin{abstract}
This paper considers the global $(l,k)$-CS problem which is the problem of controlling the system in such a way that, at least $l$ and at most $k$ processes must be in the CS at a time in the network. 
In this paper, a distributed solution is proposed in the asynchronous message-passing model.
Our solution is a versatile composition method of algorithms for $l$-mutual inclusion and $k$-mutual exclusion.
Its message complexity is $O(|Q|)$,
where $|Q|$ is the maximum size for the quorum of a coterie 
used by the algorithm, which is typically $|Q| = \sqrt{n}$.
\end{abstract}

\section{Introduction}
The mutual exclusion problem is a fundamental process synchronization problem in concurrent systems \cite{Dijkstra},\cite{mutex-survey},\cite{Review}.
It is the problem of controlling the system in such a way that no two processes execute their critical sections (abbreviated to CSs) at a time.
Various generalized versions of mutual exclusion have been studied extensively, \textit{e.g.}, $k$-mutual exclusion, mutual inclusion, $l$-mutual inclusion.
They are unified to a framework \emph{the critical section problem} in \cite{CSP}.

This paper discusses the global $(l,k)$-CS problem defined as follows.
In the entire network, the global $(l,k)$-CS problem has at least $l$ and at most $k$ processes in the CSs where $0\leq l < k\leq n$.
This problem is interesting not only theoretically but also practically. 
It is a formulation of the dynamic invocation of servers for load balancing. 
The minimum number of servers which are always invoked for quick response to
requests or for fault-tolerance is $l$. 
The number of servers is dynamically changed by system load. 
However, the total number of servers is limited by $k$ to control costs.

This paper is organized as follows.
Section \ref{SEC:Rel} reviews related works.
Section \ref{SEC:PRE} provides several definitions and problem statements.
Section \ref{SEC:ALG} provides the first solution to the global $(l,k)$-CS problem.
This solution uses a solution for the global $l$-mutual inclusion provided in Section \ref{SEC:LMUTIN} as a gadget.
In section \ref{SEC:DIS}, we discuss our concrete algorithm for the global $(l,k)$-CS problem.
In section \ref{SEC:CON}, we give a conclusion and discuss future works.

\section{Related Work}
\label{SEC:Rel}

The $k$-mutual exclusion problem is controlling the system in such a way that at most $k$ processes can execute their CSs at a time.
The $k$-mutual exclusion has been studied actively, and algorithms for this problem are proposed in, for example, 
\cite{k-coterie},\cite{bulgannawar1995distributed},\cite{chang2001generalized},\cite{SSkmutex},\cite{chaudhuri2008algorithm},\cite{Fair2008}.
However, most of them use a specialized quorum system for $k$-mutual exclusion, like $k$-coterie.

The mutual inclusion problem is the complement of the
mutual exclusion problem; unlike mutual exclusion, where at most one process is in the CS, mutual inclusion places at least one process in the CS.
Algorithms for this problem are proposed in \cite{origin-mutin} and \cite{MUTIN}.

The $l$-mutual inclusion problem is the complement of the $k$-mutual exclusion problem; $l$-mutual inclusion places at least $l$ processes in the CSs.
For this problem, to the best of our knowledge, there is no algorithm.
However, the following complementary theorem is shown in \cite{CSP}. 
\begin{theorem}\textnormal{(Complementary Theorem)}\label{comp}~
Let ${\cal A^G}_{(l,k)}$ be an algorithm for the global $(l,k)$-CS problem, Co-${\cal A^G}_{(l,k)}$ be a complement algorithm of ${\cal A^G}_{(l,k)}$, which is obtained by swapping the process states, \emph{in the CS} and \emph{out of the CS}.
Then, Co-${\cal A^G}_{(l,k)}$ is an algorithm for the global $(n-k, n-l)$-CS problem.
\end{theorem}
By this theorem, if we have an algorithm for $(n-l)$-mutual exclusion, then
we can transform it to an algorithm for $l$-mutual inclusion.
Then, Exit() (\textit{resp}. Entry()) method of $l$-mutual inclusion can make from Entry() (\textit{resp}. Exit()) method of $(n-l)$-mutual exclusion by swapping the process states.

In \cite{Algorithms}, an algorithm is proposed for the local version of $(l,k)$-CS problem.
The global CS problem is a special case of the local CS problem when the network topology is complete.
Thus, we can use the algorithm in \cite{Algorithms} as the algorithm for the global CS problem.
However, the message complexity of \cite{Algorithms} is $O(\Delta)$, where $\Delta$ is the maximum degree, as the algorithm for the local CS problem.
That is, because the maximum degree is $n$ for the global CS problem, the message complexity of \cite{Algorithms} is $O(n)$ as the algorithm for the global CS problem.

\section{Preliminary}
\label{SEC:PRE}

Let $G=(V,E)$ be a graph, where
$V=\{P_1,P_2,...,P_{n}\}$ is a set of processes and 
$E \subseteq V \times V$ is a set of bidirectional communication links
between a pair of processes.
We assume that $(P_i,P_j) \in E$ if and only if $(P_j,P_i) \in E$.
Each communication link is FIFO.
We consider that $G$ is a distributed system.
The number of processes in $G=(V,E)$ is denoted by $n (= |V|)$.
We assume that the distributed system is asynchronous, \textit{i.e.},
there is no global clock. 
A message is delivered eventually but there is no upper bound 
on the delay time and the running speed of a process may vary.

Below we present the \emph{critical section class}
which defines a common interface for
algorithms that solves a CS problem,
including $(l,k)$-CS problem, mutual exclusion, mutual inclusion, $k$-mutual exclusion and $l$-mutual inclusion.
\begin{definition}
A \emph{critical section object}, say $o$, is a distributed object (algorithm) shared by processes for coordination of accessing the critical section.
Each process has a local variable which is a reference to the object.
A class of critical section objects is called the \emph{critical section class}.
The critical section class has the following member variable and methods.
\begin{itemize*}
\item $o.{\it state}_i \in \{{\sf InCS}, {\sf OutCS}\}$ : the state of $P_i$.
\item $o$.\textnormal{Exit()} : a method to change its state from {\sf InCS} to {\sf OutCS}.
\item $o$.\textnormal{Entry()} : a method to change its state from {\sf OutCS} to {\sf InCS}.
\end{itemize*}
Each critical section object guarantees safety and liveness for accessing the critical section if critical section method invocation convention
(CSMIC), which is defined below, for object $o$ is confirmed globally.
\end{definition}
\begin{definition}
For any given process $P_i$, we say that \emph{critical section method invocation convention
(CSMIC) for object $o$ at $P_i$ is confirmed}
if and only if the following two conditions are satisfied at $P_i$.
\begin{itemize*}
\item $o$.\textnormal{Exit()} is invoked only when $o$.${\it state}_i = {\sf InCS}$ holds.
\item $o$.\textnormal{Entry()} is invoked only when $o$.${\it state}_i = {\sf OutCS}$ holds.
\end{itemize*}
\end{definition}
\begin{definition}
We say that \emph{critical section method invocation convention
(CSMIC) for object $o$ is confirmed globally}
if and only if
critical section method invocation convention
(CSMIC) for object $o$ at $P_i$ is confirmed for each $P_i \in V$.
\end{definition}

For each critical section object $o$, the vector of local states 
$(o.{\it state}_1,o.{\it state}_2,$ $\dots,o.{\it state}_n)$ of all processes forms a \emph{configuration} (global state) of a distributed system.
For each configuration $C$ for object $o$, let ${\cal CS}_o(C)$ be the set of processes $P_i$ with $o.\textit{state}_i={\sf InCS}$ in $C$.
Under each object $o$, the behaviour of each process $P_i$ is as follows, 
where we assume that, when $o.{\it state}_i$ is {\sf OutCS} (\textit{resp}. {\sf InCS}), $P_i$ eventually invokes $o$.Entry() (\textit{resp}. $o$.Exit()) and changes its state into {\sf InCS} (\textit{resp}. {\sf OutCS}).
\begin{tabbing}
\quad\=\quad\=\quad\=\quad\=\kill
\>/* $o.\textit{state}_i = (\mbox{Initial state of}~ P_i~ \mbox{in the initial configuration}) */$\\
\>$\textbf{while}~\textbf{true}~\{$\\
\>\>$\textbf{if}~(o.\textit{state}_i={\sf OutCS})~\{$\\
\>\>\>$o$.Entry();\\
\>\>\>\>/* $o.\textit{state}_i = \textsf{InCS}$ */\\
\>\>$\}$\\
\>\>$\textbf{if}~(o.\textit{state}_i={\sf InCS})~\{$\\
\>\>\>$o$.Exit();\\
\>\>\>\>/* $o.\textit{state}_i = \textsf{OutCS}$ */\\
\>\>$\}$\\
\>$\}$
\end{tabbing}

\begin{definition}(The global critical section problem).
Assume that a pair of numbers $l$ and $k$ $(0\leq l < k\leq n)$ is given on network $G=(V, E)$.
Then, an object $(l,k)$-GCS solves \emph{the global critical section problem} on $G$ if and only if the following two conditions hold in each configuration $C$.
\begin{itemize*}
\item Safety: $l\leq|{\cal CS}_{(l,k)\textnormal{-GCS}}(C)|\leq k$ at any time. 
\item Liveness: Each process $P_i\in V$ changes {\sf OutCS} and {\sf InCS} states alternately infinitely often.
\end{itemize*} 
\end{definition}
For given $l$ and $k$, we call the global CS problem as \emph{the global $(l,k)$-CS problem}.

We assume that, for object $(l,k)${\it -GCS} which is for the global $(l,k)$-CS problem, the initial configuration $C_0$ is safe, that is, $C_0$ satisfies $l\leq|{\cal CS}_{(l,k)\textnormal{-GCS}}(C_0)|\leq k$.
Note that, existing works for CS problems assume that their initial configurations are safe.
For example, for the mutual exclusion problem, most algorithms assume that each process is
in {\sf OutCS} state initially, and some algorithms (\textit{e.g.}, token based algorithms) assume that exactly one process is in {\sf InCS} state and other processes are in {\sf OutCS} state initially.
Hence our assumption for the initial configuration is a natural generalization of existing algorithms.

The typical performance measures applied to algorithms for the CS problem
are as follows.
\begin{itemize*}
\item \textit{Message complexity}: the number of message
  exchanges triggered by a pair of invocations of
  Exit() and Entry().
\item \textit{Waiting time\footnote{
     The name of this performance measure differs among previous studies and
     some (\textit{e.g.}, \cite{mutex-survey})
     refer to this performance measure as the \emph{synchronization delay}.   
   } for exit ({\it resp.} entry)}: the time period between the invocation of the Exit() ({\it resp.} Entry()) and
completion of the exit from ({\it resp.} entry to) the CS.
\item \textit{Waiting time}: the maximum one of the waiting times for exit or entry.
\end{itemize*}

Our proposed algorithm uses a coterie \cite{garcia85} 
for information exchange between processes.

\begin{definition}~\textnormal{(Coterie \cite{garcia85})}~
A \emph{coterie} ${\cal C}$ under a set $V$ is a set of subsets of $V$,
i.e., ${\cal C}=\{Q_1, Q_2, ... \}$, where $Q_i \subseteq V$ 
and it satisfies the following two conditions.
\begin{enumerate}
\item Intersection property: 
For any $Q_i,Q_j \in {\cal C}$, $Q_i \cap Q_j \ne \emptyset$ holds.
\item Minimality: 
For any distinct $Q_i,Q_j \in {\cal C}$, $Q_i \not\subseteq Q_j$ holds.
\end{enumerate}
Each member $Q_i \in {\cal C}$ is called a \emph{quorum}.
\end{definition}
We assume that, for each $P_i$, $Q_i$ is defined as a constant and is a quorum used by $P_i$. 

The algorithm proposed by
\cite{maekawa85} is a distributed mutual exclusion algorithms
that uses a coterie and it achieves a message complexity of $O(|Q|)$,
where $|Q|$ is the maximum size of the quorums in a coterie.
For example, 
the finite projective plane coterie
and the grid coterie 
achieve $|Q| = O(\sqrt{n})$, 
where $n$ is the total number of processes \cite{maekawa85}.

\section{Proposed Algorithm}
\label{SEC:ALG}

In this section, we propose a distributed algorithm for $(l,k)$-CS problem based on algorithms for $l$-mutual inclusion and $k$-mutual exclusion.
Our algorithm $(l,k)${\it -GCS} is a composition of two objects, {\it lmin} and {\it kmex}.
{\it lmin} is an algorithm for $l$-mutual inclusion, and {\it kmex} is an algorithm for $k$-mutual exclusion.
The algorithm $(l,k)${\it -GCS}
for each process $P_i \in V$ is presented in Algorithm 1.
In $(l,k)${\it -GCS}, we regards that each process state changes into {\sf OutCS} ({\it resp.} {\sf InCS}) immediately in $(l,k)${\it -GCS}.Exit() ({\it resp.} $(l,k)${\it -GCS}.Entry()), just after execution of {\it lmin}.Exit() ({\it resp.} {\it kmex}.Entry()), before execution of {\it kmex}.Exit() ({\it resp.} {\it lmin}.Entry()).
We assume that, for each $P_i$, $\mbox{$(l,k)${\it -GCS}}.{\it state}_i=\mbox{{\it lmin}}.{\it state}_i=\mbox{{\it kmex}}.{\it state}_i$ holds in the initial configuration.

In $(l,k)${\it -GCS}, safety is maintained by {\it lmin}.Exit() and {\it kmex}.Entry() because objects {\it lmin} and {\it kmex} guarantee each of their safety properties by these methods.
That is, {\it lmin}.Exit() blocks if $l$ processes are {\sf InCS}, and
{\it kmex}.Entry() blocks if $k$ processes are {\sf InCS}.

\begin{algorithm}[t]
\caption{$(l,k)${\it -GCS}}
\begin{tabbing}
\quad\=\quad\=\quad\=\quad\=\quad\=\quad\=\quad\=\quad\=\quad\=\kill
\noindent Local Variables: \\
\> {\it lmin} : \textbf{critical section object for $l$-mutual inclusion};\\
\> {\it kmex} : \textbf{critical section object for $k$-mutual exclusion};\\
\\
\noindent{Exit():} \\
\>\>\>  /* $\textit{state}_i = \textsf{InCS}$ */\\
\>  {\it lmin}.Exit(); ~/* Request */ \\
\>\>\>  /* $\textit{state}_i = \textsf{OutCS}$ */\\
\>  {\it kmex}.Exit(); ~/* Release */\\
\\

\noindent{Entry():} \\
\>\>\>  /* $\textit{state}_i = \textsf{OutCS}$ */\\
\>  {\it kmex}.Entry(); ~/* Request */ \\
\>\>\>  /* $\textit{state}_i = \textsf{InCS}$ */\\
\>  {\it lmin}.Entry(); ~/* Release */
\end{tabbing}
\end{algorithm}

\subsection{Proof of correctness of algorithm $(l,k)${\it -GCS}}
\label{SUBSEC:GCS:PROOF}

For each $P_i$, let $\#G_i$ ({\it resp.} $\#L_i, \#K_i$) be $1$ if $(l,k)${\it -GCS}.${\it state}_i={\sf InCS}$ ({\it resp.} {\it lmin}.${\it state}_i={\sf InCS}$, {\it kmex}.${\it state}_i={\sf InCS}$) holds, otherwise $0$.
Additionally, let $\#G$ ({\it resp.} $\#L, \#K$) be $\sum_{P_i} \#G_i$ ({\it resp.} $\sum_{P_i} \#L_i, \sum_{P_i} \#K_i$).
That is, $\#G={\cal CS}_{(l,k)\textnormal{-GCS}}(C)$ ({\it resp.} $\#L={\cal CS}_{{\it lmin}}(C), \#K={\cal CS}_{{\it kmex}}(C)$) in a configuration $C$.
Then, $l\leq \#L\leq n$ holds by the safety of {\it lmin}, and $0\leq \#K \leq k$ holds by the safety of {\it kmex}.
Because we assume that the initial configuration $C_0$ is safe, $l\leq \#G\leq k$ holds in $C_0$.

\begin{lemma}\label{Initial}
In the initial configuration $C_0$, {\it lmin} and {\it kmex} satisfy their safety properties.
\end{lemma}
\begin{proof}
In $C_0$, because
$(l,k)\mbox{-{\it GCS}}.{\it state}_i=\mbox{{\it lmin}}.{\it state}_i=\mbox{{\it kmex}}.{\it state}_i$ holds for each $P_i$, $\#G_i = \#L_i = \#K_i$ holds.
Hence, $\sum_{P_i} \#G_i=\sum_{P_i} \#L_i=\sum_{P_i} \#K_i$ holds.
Thus, $\#G=\#L=\#K$ holds.
Because $l\leq \#G\leq k$ holds in $C_0$, $l\leq \#L\leq k$ and $l\leq \#K \leq k$ holds in $C_0$.
Thus, {\it lmin} and {\it kmex} satisfy their safety in $C_0$.\qed
\end{proof}

\begin{lemma}\label{CSMIC}
In any execution of $(l,k)${\it -GCS}, 
CSMIC for {\it lmin} and {\it kmex} are confirmed globally.
\end{lemma}
\begin{proof}
Let $P_i$ be any process.
Because $(l,k)${\it -GCS}.${\it state}_i$ alternates by invocations of
$(l,k)${\it -GCS}.Exit() and $(l,k)${\it -GCS}.Entry(), CSMIC for $(l,k)${\it -GCS} is confirmed at $P_i$.
We show that CSMIC for {\it lmin} and {\it kmex} are also confirmed at $P_i$.
Below we show only the case of {\it lmin};
we omit the case for {\it kmex} because it is shown similarly.

First,
we show that an invariant $(l,k)${\it -GCS}.${\it state}_i = \mbox{{\it lmin}}.{\it state}_i$ holds
whenever $(l,k)${\it -GCS}.Exit() and $(l,k)${\it -GCS}.Entry() are just invoked.

In $C_0$, it is assumed that
$\mbox{$(l,k)${\it -GCS}}.{\it state}_i = \mbox{{\it lmin}}.{\it state}_i$ holds.
Hence the invariant holds.

We assume that
$(l,k)${\it -GCS}.${\it state}_i = \mbox{{\it lmin}}.{\it state}_i$ holds
when $(l,k)${\it -GCS}.Exit() and $(l,k)${\it -GCS}.Entry() are invoked.
\begin{itemize*}
\item When $(l,k)${\it -GCS}.Exit() is invoked,
  we have $\mbox{$(l,k)${\it -GCS}}.{\it state}_i = \mbox{{\it lmin}}.{\it state}_i = {\sf InCS}$
  at the beginning of invocation.
  Then, $P_i$ invokes {\it lmin}.Exit() with $\mbox{{\it lmin}}.{\it state}_i = {\sf InCS}$.
  When this invocation finishes, we have
  $\mbox{$(l,k)${\it -GCS}}.{\it state}_i = \mbox{{\it lmin}}.{\it state}_i = {\sf OutCS}$.
\item When $(l,k)${\it -GCS}.Entry() is invoked,
  we have $\mbox{$(l,k)${\it -GCS}}.{\it state}_i = \mbox{{\it lmin}}.{\it state}_i = {\sf OutCS}$
  at the beginning of invocation.
  Then, $P_i$ invokes {\it lmin}.Entry() with $\mbox{{\it lmin}}.{\it state}_i = {\sf OutCS}$.
  When this invocation finishes, we have
  $\mbox{$(l,k)${\it -GCS}}.{\it state}_i = \mbox{{\it lmin}}.{\it state}_i = {\sf InCS}$.
\end{itemize*}
Hence, any invocation of $(l,k)${\it -GCS}.Exit() and $(l,k)${\it -GCS}.Entry()
maintains the invariant.

Now, we show that
CSMIC for {\it lmin} is confirmed at $P_i$.
Because CSMIC for $(l,k)${\it -GCS} is confirmed at $P_i$,
$(l,k)${\it -GCS}.Exit() is invoked only when $(l,k)${\it -GCS}.${\it state}_i = {\sf InCS}$ holds, and
$(l,k)${\it -GCS}.Entry() is invoked only when $(l,k)${\it -GCS}.${\it state}_i = {\sf OutCS}$ holds.
Because of the invariant,
{\it lmin}.Exit() is invoked only when $\mbox{{\it lmin}}.{\it state}_i = {\sf InCS}$ holds, and
{\it lmin}.Entry() is invoked only when {\it lmin}.${\it state}_i = {\sf OutCS}$ holds.
Hence CSMIC for {\it lmin} is confirmed at $P_i$.

Because CSMIC for {\it lmin} is confirmed at $P_i$ for each $P_i$,
CSMIC for {\it lmin} is confirmed globally.
\qed
\end{proof}

\begin{lemma}
In any execution of $(l,k)${\it -GCS},
{\it lmin} and {\it kmex} satisfy safety and liveness properties.
\end{lemma}
\begin{proof}
By lemma~\ref{Initial}, in $C_0$, {\it lmin} and {\it kmex} satisfy their safety properties
because $\mbox{$(l,k)${\it -GCS}}.{\it state}_i=\mbox{{\it lmin}}.{\it state}_i=\mbox{{\it kmex}}.{\it state}_i$ holds for each $P_i$.
By lemma~\ref{CSMIC}, CSMIC for {\it lmin} and {\it kmex} are confirmed globally.
Because they are preconditions for the safety and liveness of
{\it lmin} and {\it kmex}, the lemma holds.\qed
\end{proof}

\begin{lemma}\label{LEM:GCS:SAFETY}
\textnormal{(Safety)}~
The number of processes in {\sf InCS} state is at least $l$ and at most $k$ at any time.
\end{lemma}
\begin{proof}
By the definition of $(l,k)${\it -GCS}, CSMIC for $(l,k)${\it -GCS} is confirmed globally, and $(l,k)\mbox{-GCS}.{\it state}_i =\mbox{{\it lmin}}.{\it state}_i =\mbox{{\it kmex}}.{\it state}_i$ in $C_0$, we have $\#G_i = \#L_i = \#K_i$ in $C_0$.
Thus, 
each value of $\#G_i$, $\#L_i$ and $\#K_i$ in each point of the execution is as follows.
\begin{tabbing}
\qquad\=\qquad\=\quad\=\quad\=\quad\=\quad\=\quad\=\quad\=\quad\=\kill
\noindent{$(l,k)${\it -GCS}.Exit():} \\
\>\>\>  // $(\#G_i, \#L_i, \#K_i) = (1,1,1)$\\
\>  {\it lmin}.Exit();\\
\>\>\>  // $(\#G_i, \#L_i, \#K_i) = (0,0,1)$\\
\>  {\it kmex}.Exit();\\
\>\>\>  // $(\#G_i, \#L_i, \#K_i) = (0,0,0)$\\
\\

\noindent{$(l,k)${\it -GCS}.Entry():} \\
\>\>\>  // $(\#G_i, \#L_i, \#K_i) = (0,0,0)$\\
\>  {\it kmex}.Entry(); \\
\>\>\>  // $(\#G_i, \#L_i, \#K_i) = (1,0,1)$\\
\>  {\it lmin}.Entry();\\
\>\>\>  // $(\#G_i, \#L_i, \#K_i) = (1,1,1)$
\end{tabbing}
Therefore, the following invariant $\#G_i \geq \#L_i \land \#G_i \leq \#K_i$ is satisfied.

Because $\#G = \sum_{P_i} \#G_i \geq \sum_{P_i} \#L_i = \#L$
and $\#G = \sum_{P_i} \#G_i \leq \sum_{P_i} \#K_i = \#K$ hold,
we have invariants $\#G  \geq \#L$ and $\#G \leq \#K$.
Because $\#G \geq \#L \geq l$ and $\#G \leq \#K \leq k$ holds by the safety of 
{\it lmin} and {\it kmex},
$l \leq \#G \leq k$ holds.
\qed
\end{proof}
\begin{lemma}\label{LEM:GCS:LIVENESS}
\textnormal{(Liveness)}~
Each process $P_i \in V$ alternates its state infinitely often.
\end{lemma}
\begin{proof}
By contrast, suppose that some processes do not change
\textsf{OutCS} and \textsf{InCS} states alternately infinitely often.
Let $X$ be the set of such processes.
In {\it kmex}.Exit() ({\it resp.} {\it lmin}.Entry()) method,
because $P_i$ just releases the right to be in {\sf InCS} ({\it resp.} {\sf OutCS}), the method does not block any process $P_i$ forever. 
Thus, in $(l,k)${\it -GCS}, $P_i$ is blocked only in {\it lmin}.Exit() of $(l,k)${\it -GCS}.Exit() and {\it kmex}.Entry() of $(l,k)${\it -GCS}.Entry().

Consider the case that a process $P_i\in X$ is blocked in $(l,k)${\it -GCS}.Exit() forever.
Note that, we omit the proof of the case in which $P_i$ is blocked in $(l,k)${\it -GCS}.Entry() forever because it is symmetry to the following proof.

If other processes invoke $(l,k)${\it -GCS}.Exit() and $(l,k)${\it -GCS}.Entry() alternately and complete their execution of these methods infinitely often, they complete the execution of {\it lmin}.Exit() and {\it lmin}.Entry() infinitely often.
However, because {\it lmin} satisfies its liveness, $P_i$ is not blocked forever.
Therefore, for the assumption, not only $P_i$ but also all processes must be blocked in $(l,k)${\it -GCS}.Exit() or $(l,k)${\it -GCS}.Entry() forever.
That is, $X=V$ and all processes are blocked in {\it lmin}.Exit() or {\it kmex}.Entry() forever.

Recall that it is assumed that $l \leq \#L \leq n$ holds by the safety of {\it lmin}, and $0\leq \#K \leq k$ holds by the safety of {\it kmex}.
By lemma~\ref{LEM:GCS:SAFETY}, $l\leq \#G\leq k$ holds.
If a process $P_j$ is blocked in {\it lmin}.Exit(), $\mbox{$(l,k)${\it -GCS}}.{\it state}_j=\mbox{{\it lmin}}.{\it state}_j=\mbox{{\it kmex}}.{\it state}_j={\sf InCS}$ holds, and if $P_j$ is blocked in {\it kmex}.Entry(), $\mbox{$(l,k)${\it -GCS}}.{\it state}_j=\mbox{{\it lmin}}.{\it state}_j=\mbox{{\it kmex}}.{\it state}_j={\sf OutCS}$ holds.
Therefore, $\#G=\#L=\#K$ holds.
\begin{itemize*}
\item Consider the case that all processes are blocked in {\it lmin}.Exit().
Then, $\#L=n$ holds.
However, by the assumption that {\it lmin} satisfies its safety, $l=n$ holds.
This is a contradiction because $l<k\leq n$ must hold by assumption.
\item Consider the case that there exists a process which is blocked in {\it kmex}.Entry().
By the assumption that {\it lmin} satisfies its safety, $\#L\geq l$ holds.
\begin{itemize*}
\item Consider the case that $\#L=l$ holds.
Because it is assumed that $l<k$ holds, $\#L<k$ holds, that is, $\#L=\#K<k$ holds.
Because {\it kmex} satisfies its liveness, a process which is blocked in {\it kmex}.Entry() is eventually unblocked.
This is a contradiction by the assumption that all processes are blocked forever.
\item Consider the case that $\#L>l$ holds.
Because {\it lmin} satisfies its liveness, a process which is blocked in {\it lmin}.Exit() is eventually unblocked.
This contradicts the assumption that all processes are blocked forever.\qed
\end{itemize*}
\end{itemize*}
\end{proof}

By lemmas~\ref{LEM:GCS:SAFETY} and \ref{LEM:GCS:LIVENESS}, we derived the following theorem.
\begin{theorem}\label{GCS}
$(l,k)${\it -GCS} solves the global $(l,k)$-CS problem.\qed
\end{theorem}

\section{An Example of the $l$-Mutual Inclusion}\label{SEC:LMUTIN}
Now, to show a concrete algorithm $(l,k)${\it -GCS} based on the discussion in section~\ref{SEC:ALG}, we propose a class ${\it MUTIN}(l)$ for $l$-mutual inclusion.
A formal description of the class ${\it MUTIN}(l)$
for each process $P_i \in V$ is provided in Algorithm 2.

\begin{algorithm}[hp]
\caption{A class description ${\it MUTIN}(l)$ for $l$-mutual inclusion}
\label{al1}
\begin{tabbing}
\quad\=\quad\=\quad\=\quad\=\quad\=\quad\=\quad\=\quad\=\quad\=\kill
\noindent Constants: \\
\> $Q_i : \textbf{set of processIDs}$;\\
\> $R_i : \{P_k ~|~ P_k \in V \land P_i \in Q_k \}, \textbf{set of processIDs}$;\\
\\
\noindent Local Variables: \\
\>  ${\it mx}$ : {\bf critical section object for mutual exclusion};\\ 
\>  $\textit{reqCnt}_i : \textbf{integer},~ 
                         \textbf{initially}~ 0$;\\
\>  $\textit{procsInCS}_i : \textbf{set of processIDs}$,\\
\>\>\textbf{initially}~$\{P_j\in R_i~|~{\it state}_j={\sf InCS}\}$ in a safe initial configuration;\\
\>  $\textit{currentInCS}_i : \textbf{set of processIDs}$,~\textbf{initially}~$\emptyset$;\\
\>  $\textit{ackFrom}_i : \textbf{set of processIDs},~ 
                         \textbf{initially}~ \emptyset$;\\
\>  $\textit{responseAgainTo}_i : \textbf{processID},~ 
                         \textbf{initially}~ \textbf{nil}$;\\
\>  $\textit{respAgainReqCnt}_i : \textbf{integer},~ 
                         \textbf{initially}~ 0$;\\
\\

\noindent{Exit():}\\
\>\>\>  /* $\textit{state}_i = \textsf{InCS}$ */\\
\>  {\it mx}.Entry();\\
\>   $\textit{reqCnt}_i := \textit{reqCnt}_i+1$;\\
\>  $\textit{currentInCS}_i:=\emptyset$;\\
\>  \textbf{for-each} $P_j \in Q_i$ \\
\>\>   \textbf{send} 
         $\langle \textsf{Query}, \textit{reqCnt}_i, P_i \rangle$ \textbf{to} $P_j$;\\
\>  \textbf{wait until} $(|\textit{currentInCS}_i|\geq l+1)$; \\
\>   $\textit{ackFrom}_i := \emptyset$;\\
\>  \textbf{for-each} $P_j \in Q_i$ \\
\>\>   \textbf{send} $\langle \textsf{Acquire}, P_i \rangle$
       \textbf{to} $P_j$;\\
\>  \textbf{wait until} $(\textit{ackFrom}_i= Q_i)$; \\
\>  {\it mx}.Exit();\\
\>\>\>  /* $\textit{state}_i = \textsf{OutCS}$ */\\
\\

\noindent{Entry():}\\
\>\>\>  /* $\textit{state}_i = \textsf{InCS}$ */\\
\>  \textbf{for-each} $P_j \in Q_i$ \\
\>\>   \textbf{send} $\langle \textsf{Release}, P_i \rangle$
       \textbf{to} $P_j$;
\end{tabbing}
\end{algorithm}
\setcounter{algorithm}{1}
\begin{algorithm}[hpt]
\caption{A class description ${\it MUTIN}(l)$ for $l$-mutual inclusion (continued)}
\label{al2}
\begin{tabbing}
\quad\=\quad\=\quad\=\quad\=\quad\=\quad\=\quad\=\quad\=\quad\=\kill
\noindent{On receipt of a
          $\langle \textsf{Query}, \textit{reqCnt}, P_j \rangle$ message:} \\
\>  \textbf{send} 
        $\langle \textsf{Response1}, \textit{procsInCS}_i, \textit{reqCnt}, P_i \rangle$
        \textbf{to} $\textit{P}_j$;\\
\> $\textit{responseAgainTo}_i := P_j$;\\
\> $\textit{respAgainReqCnt}_i := \textit{reqCnt}$;\\
\\
\noindent{On receipt of a $\langle \textsf{Response1}, \textit{procsInCS}, \textit{reqCnt}, P_j\rangle$ message:} \\
\> \textbf{if}~$(\textit{reqCnt}_i = \textit{reqCnt})$ \\
\>\>  $\textit{currentInCS}_i := \textit{currentInCS}_i \cup\textit{procsInCS}$; \\
\\
\noindent{On receipt of a
          $\langle \textsf{Acquire}, P_j \rangle$
          message:} \\
\>  $\textit{procsInCS}_i := \textit{procsInCS}_i \backslash \{P_j\}$;\\
\>  \textbf{send} $\langle \textsf{Ack}, P_i \rangle$ \textbf{to} $\textit{P}_j$;\\
\>  $\textit{responseAgainTo}_i := \textbf{nil}$;\\
\>  $\textit{respAgainReqCnt}_i := 0$;\\
\\
\noindent{On receipt of a
          $\langle \textsf{Ack}, P_j \rangle$
          message:} \\
\> $\textit{ackFrom}_i := \textit{ackFrom}_i \cup \{P_j \}$; \\
\\
\noindent{On receipt of a
          $\langle \textsf{Release}, P_j \rangle$
          message:} \\
\>  $\textit{procsInCS}_i := \textit{procsInCS}_i \cup \{P_j\}$;\\
\>  \textbf{if}~$(\textit{responseAgainTo}_i\neq\textbf{nil})$~\{\\
\>\>  \textbf{send} 
        $\langle \textsf{Response2}, \textit{procsInCS}_i, \textit{respAgainReqCnt}_i, P_i\rangle$ \textbf{to} $\textit{responseAgainTo}_i$;\\
\>\>  $\textit{responseAgainTo}_i := \textbf{nil}$;\\
\>\>  $\textit{respAgainReqCnt}_i := 0$;\\
\>  \}\\
\\
\noindent{On receipt of a
          $\langle \textsf{Response2}, \textit{procsInCS}, \textit{reqCnt}, P_j\rangle$ message:} \\
\> \textbf{if}~$(\textit{reqCnt}_i = \textit{reqCnt})$ \\
\>\>  $\textit{currentInCS}_i := \textit{currentInCS}_i \cup\textit{procsInCS}$; 
\end{tabbing}
\end{algorithm}

First, we present an outline how each process know the set
of processes in {\sf InCS} state in a distributed manner with quorums.
When $P_i$ changes its state,
$P_i$ notifies each process in a quorum $Q_i$ its state.
When $P_i$ wants to know the set of processes in {\sf InCS},
$P_i$ contacts with each process in $Q_i$.
For each process $P_k\in V$, because of the intersection property of quorums, there exists at least one process $P_j \in Q_k\cap Q_i\neq\emptyset$.
Hence, $P_k$ notifies its state to $P_j$, and $P_j$ sends the state of $P_k$ to $P_i$.
For this reason, when $P_i$ contacts each process in $Q_i$, $P_i$ obtains information about all the processes.

In the proposed algorithm,
each $P_i$ maintains a local variable $\textit{procsInCS}_i$ that keeps track of a set of processes in \textsf{InCS} state in $R_i$, where $R_i = \{P_k ~|~ P_k \in V \land P_i \in Q_k \}$ is the set of processes which inform about the states of processes to $P_i$.
Note that, $P_k \in R_i \Leftrightarrow P_i \in Q_k$ holds.
The value of $\textit{procsInCS}_i$ is maintained by the following way.
\begin{itemize*}
\item When $P_i$ is in {\sf InCS} state and wishes to change its state into {\sf OutCS} in Exit(),
$P_i$ sends an \textsf{Acquire} message to each $P_j \in Q_i$.
\item When $P_i$ changes its state into {\sf InCS} in Entry(),
$P_i$ sends a \textsf{Release} message to each $P_j \in Q_i$.
\item When $P_i$ receives an {\sf Acquire} message from $P_j$,
$P_i$ adds $P_j$ to $\textit{procsInCS}_i$.
\item When $P_i$ receives a {\sf Release} message from $P_j$,
$P_i$ deletes $P_j$ from $\textit{procsInCS}_i$.
\end{itemize*}
We assume that the initial value of $\textit{procsInCS}_i$ is the set of processes $P_j\in R_i$ in \textsf{InCS} state in the initial configuration.

Next, we explain the idea to guarantee safety.
When $P_i$ changes its state into {\sf InCS} by Entry(),
$P_i$ immediately sends a \textsf{Release} message to each $P_j \in Q_i$.
By Entry(), the number of processes in {\sf InCS} increases by 1.
Thus, the safety is trivially maintained.

When $P_i$ wishes to change its state into {\sf OutCS} by Exit(),
the safety is maintained by the following way.
\begin{itemize*}
\item First, $P_i$ sends a \textsf{Query} message to each process $P_j \in Q_i$.
Then, each $P_j \in Q_i$ sends a \textsf{Response1} message with ${\it procsInCS}_j$ back to $P_i$.
\item $P_i$ stores ${\it procsInCS}$ which $P_i$ received from each $P_j\in Q_i$ in variable $\textit{currentInCS}_i$.
That is, ${\it currentInCS}_i = \bigcup_{P_j \in Q_i} {\it procsInCS}_j$ holds.
\item If $|\textit{currentInCS}_i|\geq l+1$ holds, then at least $l+1$ processes are in {\sf InCS} state.
Thus, even if $P_i$ changes its state from {\sf InCS} to {\sf OutCS}, 
at least $l$ processes remain in {\sf InCS} state.
Then, safety is maintained.
Therefore, only if the condition $|\textit{currentInCS}_i|\geq l+1$ is satisfied, $P_i$ 
sends an \textsf{Acquire} message to each $P_j \in Q_i$, and changes its state to {\sf OutCS}.
\end{itemize*}

Above idea guarantees safety if only one process wishes to change its state into {\sf OutCS},
however, it does not if more than one processes wish to change their state into {\sf OutCS}.
To avoid this situation, we serialize requests which occur concurrently.
One of the typical techniques to serialize is using the priority based on the timestamp and the preemption mechanism of permissions.
This technique is employed in a lot of distributed mutual exclusion algorithms.
We use this technique for serialization, however,
for simplicity of the description of the proposed algorithm,
we use an ordinary mutual exclusion algorithm \cite{maekawa85} in the
proposed algorithm instead of explicitly use timestamp and
preemption mechanism. This is because
typical ordinary mutual exclusion algorithms use
the same mechanism for serialization,
and hence underlaying mechanism is essentially the same.
We denote the object for the ordinary mutual exclusion with {\it mx}.
When a process wishes to change its state into {\sf OutCS},
it invokes the {\it mx}.Entry() method and this allows  
it to enter the CS of mutual exclusion.
After changing its state into {\sf OutCS} successfully,
it invokes the {\it mx}.Exit() method and this allows it to
exit the CS of mutual exclusion.
Thus, by incorporating a distributed mutual exclusion algorithm {\it mx}, 
the state change from {\sf InCS} to {\sf OutCS} is serialized
between processes.
Additionally, before execution of $P_i$'s {\it mx}.Exit(), $P_i$ waits to receive {\sf Ack} messages which are responses from each $P_j\in Q_i$ to an {\sf Acquire} message sent by $P_i$.
Thus, the update of the variable $\textit{procsInCS}_j$ is atomic.
By this way, it is ensured that each process $P_k\in \textit{currentInCS}_i$ is in ${\sf InCS}$. Thus, $\#L \geq |{\it currentInCS}_i|$ is guaranteed.

Finally, we explain the idea to guarantee liveness.
When exactly $l$ processes are in {\sf InCS} state,
$P_i$ observes this by the {\sf Query}/{\sf Response1} message exchange, and
$P_i$ is blocked.
When process $P_k$ enters CS, its {\sf Release} message is sent to
each process in $Q_k$, and some $P_j \in Q_k \cap Q_i$ sends a {\sf Response2}
message to $P_i$.
Hence $P_i$ is eventually unblocked.
Note that, there exists at least such $P_j$ because of the
intersection property of quorums.

Even if there are more than $l$ processes in {\sf InCS} state,
there is a case that $P_i$ observes that the number of processes in {\sf InCS}
state is $l$ by the {\sf Query}/{\sf Response1} message exchange.
When this occurs, $P_i$ is blocked not to violate the safety.
This case occurs if the {\sf Release} message from some $P_k$ is in transit
towards $P_j \in Q_k \cap Q_i$ by asynchrony of message passing
when $P_j$ handles the {\sf Query} message from $P_i$.
Even if this case occurs,
the {\sf Release} message of $P_k$ eventually arrives to
some $P_j \in Q_k \cap Q_i$.
Then, $P_j$ sends a {\sf Response2} message to $P_i$.
Hence $P_i$ is eventually unblocked.
Because $P_i$ is unblocked by single {\sf Response2} message,
it is enough for each process to send {\sf Response2} message at most once.

Class ${\it MUTIN}(l)$
uses the following local variables for each process $P_i \in V$.
\begin{itemize*}
\item $\textit{reqCnt}_i: \textbf{integer},~ \textbf{initially}~ 0$
\begin{itemize*}
\item The request counter of $P_i$.
    This value is used by \textsf{Response1}/\textsf{Response2} message to distinguish it from the corresponding \textsf{Query} message.
\end{itemize*}
\item $\textit{procsInCS}_i: \textbf{set of processIDs}$
\begin{itemize*}
\item A set of processes in \textsf{InCS} state to the best knowledge of $P_i$.
\end{itemize*}
\item $\textit{currentInCS}_i: \textbf{set of processIDs}$
\begin{itemize*}
\item A set of processes in \textsf{InCS} state,
which are gathered by $P_i$. That is, each process in this set is known to be in \textsf{InCS} state by some process in quorum $Q_i$.
\end{itemize*}
\item $\textit{ackFrom}_i: \textbf{set of processIDs},~ \textbf{initially}~ \emptyset$
\begin{itemize*}
\item A set of processes from which $P_i$ receives an \textsf{Ack} message.
    An \textsf{Ack} message is an acknowledgment of an \textsf{Acquire} message 
    sent to each $P_j \in Q_i$, 
    where $P_i$ waits while $\textit{ackFrom}_i = Q_i$ holds.
    Due to this handshake, 
    $P_i \not\in \textit{procsInCS}_j$ is guaranteed for each $P_j \in Q_i$
    before $P_i$ invokes {\it mx}.Exit().
\end{itemize*}
\item $\textit{responseAgainTo}_i: \textbf{processID},~ \textbf{initially}~ \textbf{nil}$
\begin{itemize*}
\item A process id $P_j$ to which $P_i$ should send a \textsf{Response2} message
    when $P_j$ is waiting for $|\textit{currentInCS}_j|$ to exceed $l$. This value sets when $P_i$ receives a {\sf Query} message. 
\end{itemize*}
\item $\textit{respAgainReqCnt}_i: \textbf{integer},~ \textbf{initially}~ 0$
\begin{itemize*}
\item Request count value for the {\sf Query}
    of the process $\textit{responseAgainTo}_i$.
\end{itemize*}
\end{itemize*}

\subsection{Proof of correctness of ${\it MUTIN}(l)$}
\label{SUBSEC:GCS:PROOF2}

In this subsection, we again denote the number of processes with ${\it state}={\sf InCS}$ by $\#L$.

\begin{lemma}\label{LEM:SAFETY}
\textnormal{(Safety)}~
The number of processes in \textsf{InCS} state is at least $l$ at any time.
\end{lemma}
\begin{proof}
First, in each point of the execution, for each $P_i$,
we show that $P_j \in {\it procsInCS}_i \Rightarrow {\it state}_j = {\sf InCS}$.

In the initial configuration, ${\it procsInCS}_i$ is the set of processes in $R_i$ in {\sf InCS}.
Thus, $P_j \in {\it procsInCS}_i \Rightarrow {\it state}_j = {\sf InCS}$ holds.

Consider the case that, in the configuration such that 
$P_j \in {\it procsInCS}_i \Rightarrow {\it state}_j = {\sf InCS}$ holds, 
${\it state}_j$ changes from {\sf InCS} to {\sf OutCS}.
Such case occurs only when $P_j$ invokes Exit().
In the Exit() execution of $P_j$, because of {\it mx}.Enter()/{\it mx}.Exit() and waiting to update ${\it procsInCS}_i$ by {\sf Ack} message,
$P_j$ is not included in any ${\it procsInCS}_i$ when $P_j$ finishes the execution of Exit().
Thus, $P_j \in {\it procsInCS}_i \Rightarrow {\it state}_j = {\sf InCS}$ holds.

Now, we show that the safety is guaranteed.
In the initial configuration, it is clear that the safety is guaranteed because $\#L\geq l$.
We observe the execution after that.
In the algorithm, only when $|{\it currentInCS}_i| \geq l+1$ is satisfied, $P_i$ exits from the CS.
The value of ${\it currentInCS}_i$ is computed based on {\sf Response1} and {\sf Response2} messages. That is, ${\it currentInCS}_i = \bigcup_{P_j \in Q_i} {\it procsInCS}_j$ holds.
By {\it mx} in Exit(), because other processes than $P_i$ do not invoke Exit(),
$P_j \in {\it currentInCS}_i \Rightarrow {\it state}_j = {\sf InCS}$ holds.
Thus, $\#L \geq |{\it currentInCS}_i|$ holds.
Therefore, because $\#L \geq l+1$,
even if $P_i$ changes its state to {\sf OutCS}, $\#L \geq l$ holds.
That is, lemma holds. \qed
\end{proof}
\begin{lemma}\label{LEM:LIVENESS}
\textnormal{(Liveness)}~
Each process $P_i \in V$ 
changes \textsf{OutCS} and \textsf{InCS} states 
alternately infinitely often.
\end{lemma}
\begin{proof}
By contrast, suppose that some processes do not change
\textsf{OutCS} and \textsf{InCS} states alternately infinitely often.
Let $P_i$ be any of these processes. 
Because Entry() has no blocking operation, 
we assume that $P_i$ is blocked from executing the Exit() method.
There are three possible reasons that $P_i$ is blocked
in the Exit() method:
(1)~$P_i$ is blocked by {\it mx}.Entry(),
(2)~$P_i$ is blocked by the first \textbf{wait} statement in Exit() method, or
(3)~$P_i$ is blocked by the second \textbf{wait} statement in Exit() method.

Any process is not blocked forever by case (3) because each $P_j \in Q_i$ immediately sends back an {\sf Ack} message in response to an {\sf Acquire} message.
Below, we consider cases (1) and (2).

First, we consider the case that 
all of the blocked processes are blocked by {\it mx}.Entry(), that is, all of the blocked processes are in case (1).
However, this situation never occurs 
because we have incorporated a mutual exclusion algorithm with liveness.
Thus, at least one process is blocked in case (2).

The number of processes that is blocked in case (2) is exactly one
because no two process reach the corresponding statement at the same time
by {\it mx}.Entry().

Additionally, we claim that  
all of the processes are eventually blocked in case (1), except $P_k$ in case (2).
Each non-blocked process in {\sf InCS} state eventually calls the Exit() method
and it is then blocked by {\it mx}.Entry()
because $P_k$ obtains the lock of mutual exclusion.
Now the system reaches a configuration in which $P_k$ is blocked in case (2),
remaining $n-1$ processes are blocked in case (1), and all the processes are in {\sf InCS} state.

Finally, we show that $P_k$ is unblocked eventually.
Recall that $P_k$ is blocked in case (2), {\it i.e.},
it is waiting for a condition $|{\it currentInCS}_k| \geq l + 1$ becomes true.

The size of a collection $\bigcup_{P_j \in Q_k} {\it procsInCS}_j$, 
each of which is attached to the {\sf Response1} message sent from $P_j$ to $P_k$,
is at least $l$, {\it i.e.}, $|{\it currentInCS}_k| \geq l$ holds,
because atomic update of each ${\it procsInCS}_j$,
$\#L \geq l$ holds by the safety property, and,
for any $P_x \in V$, there exists $P_j \in Q_i$ such that $P_j \in Q_x$
by intersection property of quorums.

Although it is assumed that $|{\it currentInCS}_k| = l$ holds and
$P_k$ is blocked, a {\sf Release} message from some $P_y$ which is not in 
${\it currentInCS}_k$ eventually arrives at some $P_j$ in $Q_k$, and $P_j$ sends
a {\sf Response2} message which includes $P_y$ to $P_k$.
Note that such process $P_y$ exists because $n > l$ is assumed
and $Q_y \cap Q_k \ne \emptyset$ holds by the intersection property of quorums.
Hence, $P_k$ observes $|{\it currentInCS}_k| = l + 1$ when it receives the Response2 message, and it is unblocked.\qed
\end{proof}
\begin{lemma}\label{mess}
The message complexity of ${\it MUTIN}(l)$ is $O(|Q|)$,  
where $|Q|$ is the maximum size of the quorums of a coterie used by ${\it MUTIN}(l)$.
\end{lemma}
\begin{proof}
As noted above,
we incorporate a distributed mutual exclusion algorithm with a message complexity of $O(|Q|)$, 
such as that proposed by \cite{maekawa85}.
Thus, {\it mx} requires $O(|Q|)$ messages.

In the Exit() method, $P_i$ sends $|Q_i|$ \textsf{Query} messages.
For each $P_j \in Q_i$, $P_j$ sends exactly one \textsf{Response1} message
for each \textsf{Query} message: $|Q_i|$ \textsf{Response1} messages.
$P_i$ sends $|Q_i|$ \textsf{Acquire} messages.
Then, each $P_j \in Q_i$ sends an \textsf{Ack} message: 
$|Q_i|$ \textsf{Ack} messages.
Hence, $O(|Q|)$ messages are exchanged.

In the Entry() method, 
$P_i$ sends $|Q_i|$ \textsf{Release} messages.
For each $P_j \in Q_i$, $P_j$ sends at most one \textsf{Response2} message
for \textsf{Query} messages: $|Q_i|$ \textsf{Response2} messages.
Therefore, $O(|Q|)$ messages are exchanged.

In total, $O(|Q|)$ messages are exchanged.\qed
\end{proof}
\begin{lemma}\label{wait}
The waiting time of ${\it MUTIN}(l)$ is 7.
\end{lemma}
\begin{proof}
The waiting time is 3 for the mutual exclusion algorithm employed by ${\it MUTIN}(l)$, which was described 
by Maekawa~\cite{maekawa85}
(2 for Entry() and 1 for Exit(); 
see \cite{mutex-survey}.) 

In Exit(),
a chain of messages, i.e.,
\textsf{Query}, \textsf{Response1}, \textsf{Acquire}, \textsf{Ack}
is exchanged between $P_i$ and the processes in $Q_i$.
Hence, 4 additional time units are required.
In total, the waiting time for exit is 7 time units.

In Entry(), a \textsf{Release} message and a \textsf{Response2} message
are exchanged between $P_i$ and the processes in $Q_i$.
The waiting time for entry is 2 time units.

Thus, the waiting time is 7 time units.\qed
\end{proof}

By lemmas~\ref{LEM:SAFETY}-\ref{wait}, we derived the following theorem.
\begin{theorem}
${\it MUTIN}(l)$ solves the $l$-mutual inclusion problem with a message complexity of $O(|Q|)$
where $|Q|$ is the maximum size of the quorums of a coterie used by ${\it MUTIN}(l)$.
The waiting time of ${\it MUTIN}(l)$ is 7.\qed
\end{theorem}

\section{Discussion}
\label{SEC:DIS}

In this section, finally, we discuss the case that, by the complementary theorem (theorem~\ref{comp}), in $(l,k)${\it -GCS},
we use the proposed class as ${\it MUTIN}(l)$ to obtain object {\it lmin}
and as ${\it MUTIN}(n-k)$ to obtain object {\it kmex}.

Then, by the proof of lemma~\ref{mess}, the message complexity is $O(|Q|)$. 
Additionally, by the proof of lemma~\ref{wait}, both of waiting times for exit and entry of $(l,k)${\it -GCS} are 9. 
Thus, by theorem~\ref{GCS}, we derive the following theorem.
\begin{theorem}
$(l,k)${\it -GCS} solves the global $(l,k)$-CS problem
with a message complexity of $O(|Q|)$
where $|Q|$ is the maximum size of the quorums of a coterie used by $(l,k)${\it -GCS}.
The waiting time of $(l,k)${\it -GCS} is 9.\qed
\end{theorem}

\section{Conclusion}
\label{SEC:CON}
In this paper, we discuss the global critical section problem in asynchronous message passing distributed systems.
Because this problem is useful for fault-tolerance and load
balancing of distributed systems, we can consider various future applications.

In the future, we plan to perform extensive simulations and confirm the performance of our algorithms under various application scenarios.
Additionally, we plan to design a fault tolerant algorithm for the problem.

\section*{Acknowledgement}
This work is supported in part by KAKENHI No. 16K00018 and 26330015.

\bibliographystyle{elsarticle-num}
\bibliography{lkmutin}

\end{document}